\documentclass{amsart} [12 pt]
\usepackage{amsmath}
\usepackage{graphicx}
\usepackage{graphicx,xcolor}
\usepackage{amsfonts} 
\usepackage{amssymb}
\usepackage{pdfsync}
\usepackage{color}

\usepackage{epsfig}
\theoremstyle{plain}
\newtheorem{theorem}{Theorem}[section]

\newtheorem{lemma}[theorem]{Lemma}
\newtheorem{proposition}[theorem]{Proposition}

\newtheorem{remark}[theorem]{Remark}

\theoremstyle{definition}
\theoremstyle{remark}
\numberwithin{equation}{section}

\newcommand{\ep}{\varepsilon}

\newcommand{\ffi}{\varphi}

\newcommand{\R}{\mathbb{R}}

\newcommand{\PP}{\mathbb{P}}
\newcommand{\EE}{\mathbb{E}}

\newcommand{\T}{\mathcal{T}}
\newcommand{\LL}{\mathcal{L}}

\newcommand{\mbf}{\mathbf}

\newcommand{\h}{\tilde h_\ep}

\usepackage{latexsym}
 
\title[A diffusion limit for a test particle in a random distribution of scatterers]{A diffusion limit for a test particle in a random distribution of scatterers}
\author[G. Basile]
{G. Basile}
\address[Giada Basile]{Dipartimento di Matematica ``Guido Castelnuovo'', Sapienza Universit\`a di Roma, P.le Aldo Moro 5, I-00185 Roma, Italy}
\email[G. Basile]{basile@mat.uniroma1.it}
\author[A. Nota]
{A. Nota}
\address[Alessia Nota]{Dipartimento di Matematica ``Guido Castelnuovo'', Sapienza Universit\`a di Roma, P.le Aldo Moro 5, I-00185 Roma, Italy}
\email[A. Nota]{nota@mat.uniroma1.it}
\author[M. Pulvirenti]
{M. Pulvirenti}
\address[Mario Pulvirenti]{Dipartimento di Matematica ``Guido Castelnuovo'', Sapienza Universit\`a di Roma, P.le Aldo Moro 5, I-00185 Roma, Italy}
\email[M. Pulvirenti]{pulvirenti@mat.uniroma1.it}

\begin{document}
\begin{abstract}
We consider a point particle moving in a random distribution of obstacles
described by a potential barrier.
We show that,  in a weak-coupling regime, under a diffusion limit
suggested by the potential itself, the probability distribution of
the particle converges to the solution of the heat equation.  The
diffusion coefficient is given by the Green-Kubo formula associated to
the generator of the diffusion process dictated by the linear Landau
equation.
\end{abstract}

\maketitle
\section{Introduction}
The evolution of the density of a test particle moving in a configuration of obstacles
is described at mesoscopic level by linear kinetic equations. 
They are obtained from the microscopic Hamiltonian
dynamics under a kinetic scaling of space and time, namely $t\to \ep t$, $x\to \ep x$ and a suitable rescaling of the density of the obstacles and the intensity of the interaction.
Accordingly to the resulting frequency of collisions, the mean free path of the particle can have or not macroscopic length and different kinetic equations arise.
Typical examples are the linear Boltzmann equation and the linear Landau equation.

The first rigorous result appeared  in 1969 in the paper of  Gallavotti \cite{G}, who derived a linear Boltzmann equation starting from 
a random distribution of fixed hard scatterers 
  in the Boltzmann-Grad limit (low density), namely when the number of collisions is small, thus the mean free path of 
the particle is macroscopic.  The result was improved by Spohn \cite{S}.
 
In the weak-coupling regime,
when there are very many but weak collisions, a linear Landau equation appears
\begin{equation}\label{Landau}
(\partial_{t}+v\cdot\nabla_{x})f(x,v,t)=B \Delta_{|v|} f(x,v,t),
\end{equation}
where $\Delta_{|v|}$ is the Laplace-Beltrami operator on the $d$-dimensional sphere 
of radius $|v|$. It describes a momentum diffusion, i.e. the velocity process is a Brownian motion on the (kinetic) energy sphere. This intuitively follows from the facts that 
there are many elastic 
  collisions with obstacles isotropically distributed. 
The diffusion coefficient $B$ is proportional to the variance of the transferred momentum in a single collision and depends on the shape of the interaction potential.
 The first  result in this direction was obtained by
Kesten and Papanicolau in 1978 for a particle in $\R^3$ and by D\"urr, Goldstein and Lebowitz in 1987
for a particle in $\R^2$ for sufficiently smooth interaction potentials.

The linear Landau equation yields also in an intermediate scale between low density and weak-coupling regime, namely when the (smooth) interaction potential $\phi$ rescales according to 
$\phi\to \ep^\alpha \phi$, $\alpha\in (0,1/2)$ 
and  the density of the obstacles is of order $\ep^{-2\alpha-(d-1)}$ (\cite{DR}, \cite{K}). 
The limiting cases $\alpha=0$ and $\alpha=1/2$ correspond respectively to the low density limit and the weak-coupling limit.

In the present paper we want to investigate the limit $\ep\rightarrow 0$ in the intermediate case, namely when 
$\alpha>0$ but sufficiently small, for an interaction potential no more smooth given by a circular potential barrier, in dimension two. 
The physical interest of this problem is connected to the geometric optics since the trajectory 
of the test particle is that of a light ray traveling in a medium (say water) in presence of circular drops of a different substance with smaller refractive index (say air). The opposite situation, namely drops of water in a medium of air, can be described as well by the circular well potential.
Our analysis applies also to this case with minor modifications, but we consider only the case
of potential barrier for sake of concreteness.

The novelty of this choice is that in this case the diffusion coefficient $B$ 
diverges logarithmically. Roughly speaking, the
asymptotic equation for the density  of the Lorentz particle reads
\begin{equation}\label{approxLandau}
(\partial_{t}+v\cdot\nabla_{x})f(x,v,t)\sim\,|\log\ep| \,B \Delta_{|v|} f(x,v,t),
\end{equation}
which suggests to look at a longer time scale  $t\to |\log\ep|t$. As expected,  a  diffusion in space arises.

The proof  follows the original constructive idea, due to Gallavotti \cite{G},
for the low-density limit of a hard-sphere system. This approach is based on a suitable change of variables which leads to a Markovian approximation described by a linear Boltzmann equation. This presents some technical difficulties since some of the random configurations lead to trajectories that ``remember'' too much preventing the Markov property of the limit. In the two-dimensional case the probability of those bad behaviors producing memory effects (correlation between the past and the present) is nontrivial. Thus we need to control the unphysical trajectories: we estimate explicitly the set of bad configurations of the scatterers (such as the set of configurations yielding recollisions or interferences) showing that it is negligible in the limit (see \cite{DP}). The control of memory effects still holds for a longer time scale $|\log\ep|$ which allows to get the heat equation from the rescaled linear Boltzmann equation. \\
\indent We remark that the diffusive limit analyzed in the present paper is suggested by the divergence of the diffusion coefficient for the particular choice of the potential we are considering. However the same techniques could work in presence of a smooth, radial, short-range potential $\phi$. Also in this case we obtain a diffusive equation as longer time scale limit of a linear Boltzmann equation (Section 5). This is in the same spirit of \cite{KR} and \cite{LE}.

\section{Main results}
Consider a point particle of mass one in $\R^2$, moving in a random distribution of fixed scatterers whose center are denoted by $c_1,\dots,c_N\in\R^2$. The equation of motion are
\begin{equation}\label{eqmot}
\left\{\begin{array}{ll}
\dot{x}=v&\\
\dot{v}=-\sum_{i=1}^{N}\nabla\phi(|x-c_i|)&,
\end{array}\right.
\end{equation}
where $(x,v)$ denote position and velocity of the test particle, $t$ the time and, as usual, $\dot{A}=\frac{\,dA}{\,dt}$ indicates the time derivative for any time dependent variable $A$.\\
Finally $\phi:\R^{+}\to\R$ is a given spherically symmetric potential.\\
\indent To outline a kinetic behavior of the particle, we usually introduce a scale parameter $\ep>0$, indicating the ratio between the macroscopic and the microscopic variables, and rescale according to 
\begin{equation*}
x\rightarrow\ep x,\; t\rightarrow\ep t,\; \phi\rightarrow\ep^{\alpha}\phi
\end{equation*}
with $\alpha\in[0,1/2]$. Then Eq.ns \eqref{eqmot} become
\begin{equation}\label{scaled}
\left\{\begin{array}{ll}
\dot{x}=v&\\
\dot{v}=-\ep^{\alpha-1}\sum_{i}\nabla\phi(\frac{|x-c_i|}{\ep})&.
\end{array}\right.
\end{equation}
\indent We assume the scatterers  $\mbf{c}_{N}=(c_1,\dots,c_N)$ distributed according to a Poisson distribution of intensity $\mu_{\varepsilon}=\mu\varepsilon^{-\delta}$, where $\delta=1+2\alpha$. This means that the probability density of finding $N$ obstacles in a bounded measurable set $\Lambda\subset\R^2$ is given by
\begin{equation}\label{poisson}
\PP_{\ep}(\,d\mbf{c}_{N})=e^{-\mu_{\ep}|\Lambda|}\frac{\mu_{\varepsilon}^{N}}{N!}\,dc_1,\dots,\,dc_N
\end{equation}
where $|\Lambda|=\text{meas}\Lambda$.\\
\indent Now let $T^t_{\mbf{c}_{N}}(x,v)$ be the Hamiltonian flow solution of Eq.n \eqref{scaled} with initial datum $(x,v)$ in a given sample $\mbf{c}_{N}=(c_1,\dots,c_N)$ of obstacles (skipping the $\ep$ dependence for notational simplicity) and, for a given initial probability distribution $f_0=f_0(x,v)$, consider the quantity
\begin{equation}\label{valatt}
f_{\ep}(x,v,t)=\EE_{\ep}[f_0(T^{-t}_{\mbf{c}_{N}}(x,v))],
\end{equation}
where $\EE_{\ep}$ is the expectation with respect to the measure $\PP_{\ep}$ given by \eqref{poisson}.\\
\indent In the limit $\ep\to 0$ we expect that the probability distribution \eqref{valatt} solves a linear kinetic equation depending on the value of $\alpha$. More precisely if $\alpha=0$ (low-density or Boltzmann-Grad limit) then $f_{\ep}$ converges to $f$, the solution of the following linear Boltzmann equation 
\begin{equation}\label{boltzmann}
(\partial_{t}+v\cdot\nabla_{x})f(x,v,t)=\text{L}f(x,v,t)
\end{equation}
where 
\begin{equation}
\text{L}f(x,v,t)=\mu|v|\int_{-1}^{1}\,d\rho\{f(v')-f(v)\}
\end{equation}
and where 
\begin{equation}
v'=v-2(\omega\cdot v)\omega.
\end{equation}
Here we are assuming $\phi$ of range one i.e. $\phi( r)=0$ if $r>1$, and $\omega=\omega(\rho,|v|)$ is the unit vector obtained by solving the scattering problem associated to $\phi$. This result was proven and discusses in \cite{BBS},\cite{DP},\cite{G},\cite{S}. \\
On the other hand, if $\alpha=1/2$, the corresponding limit, called weak-coupling limit, yields the linear Landau equation (see \cite{DGL} and \cite{K})
\begin{equation}\label{Landau}
(\partial_{t}+v\cdot\nabla_{x})f(x,v,t)=\LL f(x,v,t)
\end{equation}
where
\begin{equation}
\LL f(v)=B \Delta_{|_{S_{|v|}}}, 
\end{equation} 
and 
\begin{equation}
\label{coeff diff}
B=\frac{\pi\mu}{|v|}\int_{0}^{\infty} r^2\hat{\phi}( r)^2\,dr.
\end{equation}
\noindent
Note that $\hat{\phi}$ is real and spherically symmetric.



In the present paper we want to investigate the limit $\ep\rightarrow 0$, in case $\alpha>0$ sufficiently small, when the diffusion coefficient $B$ given by \eqref{coeff diff} is diverging. 
Actually we consider the specific example 
\begin{equation}
\label{pot barr}
\phi( r)=\left\{\begin{array}{ll}
1\quad \text{if}\; r<1&\\
0\quad \text{otherwise}&,
\end{array}\right.
\end{equation}
namely a circular potential barrier.\\
\indent 
For a   potential of the form \eqref{pot barr}
a simple computation  shows that $B$ defined in \eqref{coeff diff} diverges logarithmically. Therefore we are interested in characterizing the asymptotic behavior of $f_{\ep}(x,v,t)$, given by \eqref{valatt}, under the scaling illustrated above.
The main result of the present paper can be summarized  in the following theorem.

\begin{theorem}\label{th:main th}
Suppose $f_0\in C_0(\R^2\times\R^2)$ a continuous, compactly supported initial probability density. Suppose also that $|D_{x}^kf_0|\leq C$, where $D_x$ is any partial derivative with respect to $x$ and $k=1,2$.
Finally assume $\alpha\in (0,1/8)$. The following statements hold
\begin{enumerate}
\item[1)] if $\mu_{\ep}=\ep^{-2\alpha-1}$, for all $t\in (0,T]$,  $T>0$,
\begin{equation*}
\lim_{\ep\to 0}f_{\ep}(x,v,t)=\langle f_0\rangle:=\frac 1{2\pi}\frac 1{|v|}\int_{S_{|v|}} 
{f_0(x,v)\,dv}.
\end{equation*}
The convergence is in $L^2(\R^2\times S_{|v|} )$.

\item[2)] if $\mu_{\ep}=\frac{\ep^{-2\alpha-1}}{|\log\ep|}$, for all $t\in (0,T]$, $T>0$,  
\begin{equation*}
\lim_{\ep\to 0}f_{\ep}(x,v,t)=f(x,v,t),
\end{equation*}
where $f$ solves the Landau equation \eqref{Landau} with a renormalized diffusion coefficient 
\begin{equation}
\label{diff coeff}
B:=\lim_{\ep\to 0}{\frac{\mu_\ep}{2}\ep\,|v|\int_{-1}^{1}{\theta^2(\rho)\,d\rho}}.
\end{equation}
The convergence is in $L^2(\R^2\times S_{|v|} )$. \\

\item[3)] if $\mu_{\ep}=\ep^{-2\alpha-1}$, defining $F_{\ep}(x,v,t):=f_{\ep}(x,v,t |\log\ep|)$, for all $t\in [0,T)$, $T>0$, 
\begin{equation*}
\lim_{\ep\to 0}F_{\ep}(x,v,t)=\rho(x,t), 
\end{equation*}
where $\rho$ solves the following heat equation
\begin{equation}
\left\{
 \begin{array}{l}\vspace{0.2cm}
\partial_t\varrho=D\Delta\varrho\\ 
\varrho(x,0)=\langle f_0\rangle, 
\end{array} \right.
\end{equation}
with $D$ given by the Green-Kubo formula
\begin{equation}\label{GK}
D=\frac{1}{\mu}|v|\int_{S_{|v|}}{v\cdot\big(-\Delta_{|_{S_{|v|}}}^{-1}\big)v\,dv}=\frac{2\pi}{\mu}|v|^2 \int_{0}^{\infty}\EE\big[v\cdot v(t,v)\big]\,dt, 
\end{equation}
where $v(t,v)$ is the stochastic process dictated by the generator of the Landau equation starting from $v$ and $\EE[\cdot]$ denotes the expectation  with respect to the invariant measure 
on $S_{|v|}$.
The convergence is in $L^2(\R^2\times S_{|v|} )$.
\end{enumerate}
\end{theorem}

Some comments to Theorem \ref{th:main th} are in order. As we shall prove in Section 4, the asymptotic behavior of the mechanical system we are considering is the same as the Markov process ruled by the linear Landau equation with a diverging factor in front of $\LL$. This is equivalent to consider the  limit in the Euler scaling  of the  linear Landau equation, which is trivial. The system quickly thermalizes to the local equilibrium just given by $\langle f_0\rangle$. This is point 1).

To detect something non-trivial we have to exploit longer times in which the local equilibrium starts to evolve (according to the diffusion equation), see point 3).
Note however that, rescaling differently the density of the Poisson process, we can recover the kinetic picture given by Landau equation (with a renormalized diffusion coefficient $B$) as in \cite{DR}, see point 2).

We finally remark that this picture is made possible because the recollisions set (see below for the precise definition) is negligible, as established in Section 4.
We believe that the present result could be recovered also in high-density regimes $\alpha\in\big(\frac1 8,\frac 1 2\big]$, namely also when the recollisions are not negligible anymore. However in this case different ideas and techniques are indeed necessary.

The plan of the paper is the following. In the next Section we illustrate our strategy and establish some preliminary results. In Section 3 we prove Theorem 1.1. Finally in Section 4 we prove a basic Lemma showing that our non-Markovian system can indeed be approximated by a Markovian one, easier to handle with.

\section{Strategy}
We follow the explicit approach in {\cite{G}}, {\cite{DP}} and {\cite{DR}}.\\
By  \eqref{valatt} we have, for $(x,v)\in\R^2\times\R^2$, $t>0$, 
\begin{equation}
\label{formula1}
f_{\ep}(x,v,t)=e^{-\mu_{\ep}|B_t(x,v)|}\sum_{N\geq 0}\frac{\mu_{\ep}^{N}}{N!}\int_{B_t(x,v)^N}d\mbf{c}_{N}\, f_0(T^{-t}_{\mbf{c}_{N}}(x,v))
\end{equation}
where $T^t_{\mbf{c}_{N}}(x,v)$ is the Hamiltonian flow generated by the Hamiltonian
\begin{equation}
\frac{1}{2}v^2+\ep^{\alpha}\sum_{j}\phi\left(\frac{|x-c_j|}{\ep}\right)
\end{equation}
where $\phi$ is given by \eqref{pot barr}, and initial datum $(x,v)$. Finally $B_t(x,v)=B(x,|v|t)$, where here and in the following, $B(x,R)$ denotes the disk of center $x$ and radius $R$.\\
The explicit solution to the equation of motion is obtained by solving the single scattering problem by using the energy and angular momentum conservation (see figure below).

\begin{figure}[ht]
\centering
\includegraphics[scale= 0.8]{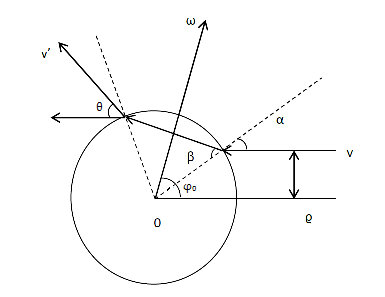}
\caption{Scattering}\label{fig:1}
\end{figure}

Here we represent the scattering of a particle entering in the ball 
\begin{equation*}
B(0,1)=\{x\;\text{s.t.}\; |x|<1\}
\end{equation*} 
toward a potential barrier of intensity $\phi(x)=\ep^{\alpha}$.

We have an explicit expression for the refractive index 
\begin{equation}
\label{refractive index }
n_{\ep}=\frac{\sin\alpha}{\sin\beta}=\frac{|\bar{v}|}{|v|}=\sqrt{1-\frac{2\ep^{\alpha}}{v^2}},
\end{equation}
where $v$ is the initial velocity, $\bar v$ the velocity inside the barrier, $\alpha$ the angle of incidence and $\beta$ the angle of refraction.
The scattering angle is $\Theta=\pi-2\ffi_0=2(\beta-\alpha)$ and the impact parameter is $\rho=\sin\alpha$. (See Appendix \ref{scatt} for a detailed analysis of the scattering problem.)

\begin{figure}[h]
\centering
\includegraphics[scale=0.7]{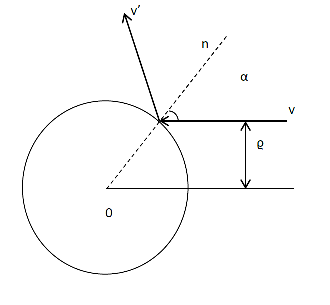}
\caption{Elastic reflection}\label{fig:3}
\end{figure}

\begin{remark}Formula \eqref{refractive index } makes sense if $\frac{2\ep^{\alpha}}{v^2}<1$ or $\rho=\sin\alpha<\sqrt{1-\frac{2\ep^{\alpha}}{v^2}}$.\\
When one of such two inequalities is violated, the outgoing velocity is the one given by the elastic reflection.
\end{remark}


\begin{figure}
\centering
\includegraphics[scale=0.7]{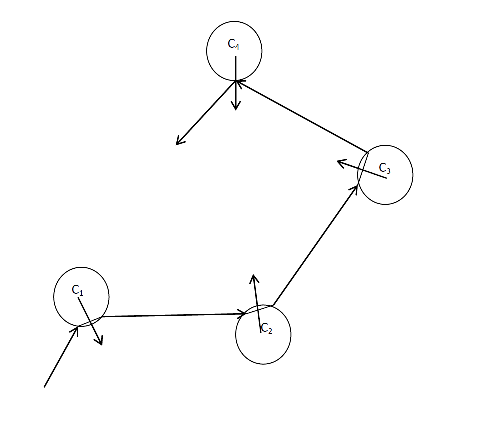}
\caption{A typical trajectory}\label{fig:2}
\end{figure}

\noindent After the scaling 
\begin{equation*}
x\rightarrow\ep x,\; t\rightarrow\ep t
\end{equation*}
the scattering process takes place in a disk of radius $\ep$, but the velocities (and hence the angles) are invariant. A picture of a typical trajectory is given as in Figure 3. Here we are not considering possible overlappings of obstacles. The scattering process can be solved in this case as well. However, as we shall see, this event is negligible because  of the moderate densities we are considering.\\
\indent Coming back to Eq.n \eqref{formula1}, we distinguish the obstacles of the configuration $\mbf{c}_{N}=c_1\dots c_N$ which, up to the time $t$, influence the motion, called internal obstacles, and the external ones. More precisely $c_i$ is internal if
\begin{equation}
\inf_{-t\leq s\leq 0}|x_{\ep}(s)-c_i|<\ep,
\end{equation}
while $c_i$ is external if 
\begin{equation}
\inf_{-t\leq s\leq 0}|x_{\ep}(s)-c_i|\geq\ep.
\end{equation}
Here $(x_{\ep}(s),v_{\ep}(s))=T_{\mbf{c}}^{s}(x,v)$.

Note that the integration over the external obstacles can be done so that
\begin{equation}
\begin{split}
\label{formula2}
f_{\ep}(x,v,t)&=\sum_{Q\geq 0}\frac{\mu_{\ep}^{Q}}{Q!}\int_{B_t(x,v)^Q}d\mbf{b}_{Q}\, e^{-\mu_{\ep}|\T(\mbf{b}_{Q})|}f_0(T^{-t}_{\mbf{b}_{Q}}(x,v))\\&
\chi(\{\text{the}\;\mbf{b}_{Q}\;\text{are}\;\text{internal}\}).
\end{split}
\end{equation}
Here and in the sequel $\chi(\{\dots\})$ is the characteristic function of the event $\{\dots\}$.\\
Moreover $\T(\mbf{b}_{Q})$ is the tube:
\begin{equation}
\T(\mbf{b}_{Q})=\{y\in B_t(x,v)\;\text{s.t.}\;\exists s\in(-t,0)\;\text{s.t.}\;|y-x_{\ep}(s)|<\ep\}.
\end{equation}
Note that 
\begin{equation}
|\T(\mbf{b}_{Q})|\leq 2\ep|v|t.
\end{equation}

\begin{figure}
\centering
\includegraphics{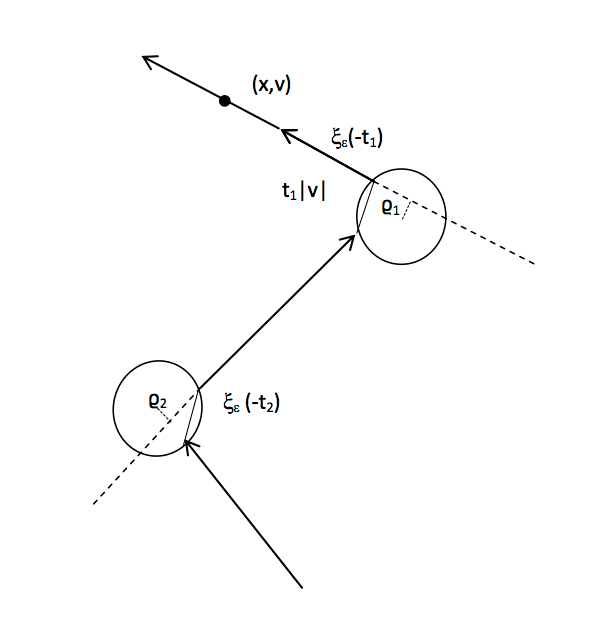}
\caption[size=0.7]{The change of variables}\label{fig:4}
\end{figure}

Instead of considering $f_{\ep}$ we introduce
\begin{equation}
\begin{split}
\label{formula3}
\tilde{f}_{\ep}(x,v,t)&=e^{-2\ep^{-2\alpha}|v|t}\sum_{Q\geq 0}\frac{\mu_{\ep}^{Q}}{Q!}\int_{B_t(x,v)^Q}d\mbf{b}_{Q}\\&
\chi(\{\text{the}\;\mbf{b}_{Q}\;\text{are}\;\text{internal}\})\chi_1(\mbf{b}_{Q})f_0(T^{-t}_{\mbf{b}_{Q}}(x,v)).
\end{split}
\end{equation}
where 
\begin{equation}
\chi_1(\mbf{b}_{Q})=\chi\{\mbf{b}_{Q}\;\text{s.t.}\; b_i\notin B(x,\ep)\;\text{and}\; b_i\notin B(x(-t),\ep)\;\text{for}\;\text{all}\; i=1,\dots,Q\}
\end{equation}
Obviously
\begin{equation}
f_{\ep}\geq \tilde{f}_{\ep}.
\end{equation}

Following \cite{G},\cite{DP},\cite{DR}
we would like to perform the following change of variables
\begin{equation*}
0\leq t_1<t_2<\dots<t_Q\leq t
\end{equation*}
\begin{equation}
\label{change var}
b_1,\dots,b_Q\rightarrow \rho_1,t_1,\dots,\rho_Q,t_Q
\end{equation}
where, after ordering the obstacles $b_1,\dots,b_Q$ according to the scattering sequence, $\rho_i$ and $t_i$ are the impact parameter and the entrance time of the light particle in the protection disk around $b_i$.

More precisely, fixed an impact parameter $\rho$ and an entrance time $t$ we construct $b=b(\rho,t)$, the center of the obstacle. Then we perform the backward scattering and iterate the procedure to construct a trajectory $(\xi_{\ep}(s),\eta_{\ep}(s))$.
However $(\xi_{\ep}(s),\eta_{\ep}(s))=(x_{\ep}(s),v_{\ep}(s))$ (therefore the mapping \eqref{change var} is one-to-one) only outside the following pathological situations.\\
i) \textbf{Overlapping}.\\ If $b_i$ and $b_j$ are both internal then $B(b_i,\ep)\cap B(b_j,\ep)\neq\emptyset$ .\\
ii) \textbf{Recollisions}.\\ There exists $b_i$ such that for $s\in(t_{j},t_{j+1})$, $j>i$, $\xi_{\ep}(-s)\in B(b_i,\ep)$. \\
iii) \textbf{Interferences}.\\ There exists $b_i$ such that $\xi_{\ep}(-s)\in B(b_i,\ep)$ for $s\in(t_{j},t_{j+1})$, $j<i$.\\

\begin{figure}
\centering
\includegraphics[scale=0.8]{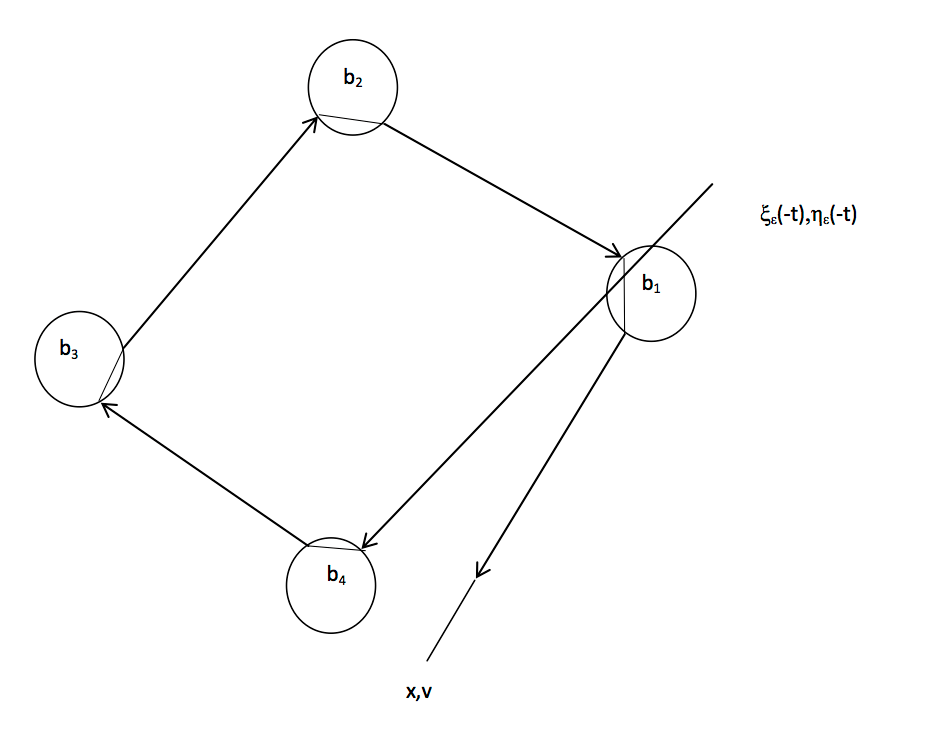}
\caption[scale=0.7]{Recollisions}\label{fig:5}
\end{figure}

\begin{figure}[htbp]
\centering
\includegraphics[scale=0.8]{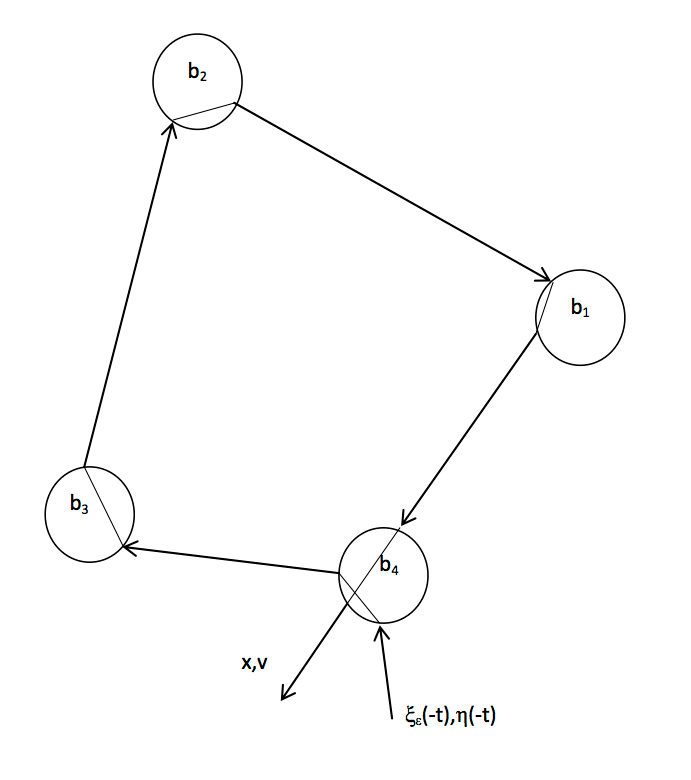}
\caption[scale=0.5]{Interferences}\label{fig:6}
\end{figure}

We simply skip such events by setting
\begin{equation*}
\chi_{ov}=\chi(\{\mbf{b}_Q\;\text{s.t.}\;\text{i)}\;\text{is}\;\text{realized}\}),
\end{equation*}
\begin{equation*}
\chi_{rec}=\chi(\{\mbf{b}_Q\;\text{s.t.}\;\text{ii)}\;\text{is}\;\text{realized}\}),
\end{equation*}
\begin{equation*}
\chi_{int}=\chi(\{\mbf{b}_Q\;\text{s.t.}\;\text{iii)}\;\text{is}\;\text{realized}\}),
\end{equation*}
and defining
\begin{equation}
\begin{split}
\label{formula4}
\bar{f}_{\ep}(x,v,t)&=e^{-2\ep^{-2\alpha}|v|t}\sum_{Q\geq 0}\mu_{\ep}^{Q}\int_{0}^{t}dt_1\dots\int_{0}^{t_{Q-1}}dt_Q\\&
\int_{-\ep}^{\ep}d\rho_1\dots\int_{-\ep}^{\ep}d\rho_Q\,\chi_1(1-\chi_{ov})(1-\chi_{rec})(1-\chi_{int})f_0(\xi_{\ep}(-t),\eta_{\ep}(-t)).
\end{split}
\end{equation}
Note that $\bar{f}_{\ep}\leq\tilde{ f}_{\ep}\leq f_{\ep}$. Note also that in \eqref{formula4} we have used the change of variables \eqref{change var} for which, outside the pathological sets i), ii), iii), $T^{-t}_{\mbf{b}_Q}(x,v)=(x_{\ep}(-t),v_{\ep}(-t))$.\\
\indent Next we remove $\chi_1(1-\chi_{ov})(1-\chi_{rec})(1-\chi_{int})$ by setting 
\begin{equation}
\begin{split}
\label{formula5}
\bar{h}_{\ep}(x,v,t)&=e^{-2\ep^{-2\alpha}|v|t}\sum_{Q\geq 0}\mu_{\ep}^{Q}\int_{0}^{t}dt_1\dots\int_{0}^{t_{Q-1}}dt_Q \\&
\int_{-\ep}^{\ep}d\rho_1\dots\int_{-\ep}^{\ep}d\rho_Q\, f_0(\xi_{\ep}(-t),\eta_{\ep}(-t)).
\end{split}
\end{equation}
\indent We can prove:
\begin{proposition}\label{th:prop}
\begin{equation*}
\bar{f}_{\ep}(t)=\bar{h}_{\ep}(t)+\ffi_1(\ep,t)
\end{equation*}
where $\|\ffi_1(\ep,t)\|_{L^1}\to 0$ as $\ep\to 0$ for all $t\in[0,T].$
\end{proposition}
\begin{remark}\label{th:remarkprop}
Proposition \ref{th:prop} still holds for longer times, namely:
\begin{equation*}
\|\ffi_1(\ep,t)\|_{L^1}\underset{\ep\to 0}{\longrightarrow} 0\quad\forall t\in[0,|\log\ep|T],\quad T>0.
\end{equation*}
\end{remark}
We postpone the proof of the above Proposition in the last Section.\\
Next we consider the limiting trajectory $\bar{\xi}_{\ep}(s),\bar{\eta}_{\ep}(s)$ obtained by considering the collision as instantaneous.\\
More precisely, for the sequence $t_1,\dots,t_Q$ $\rho_1,\dots\rho_Q$ consider the sequence $v_1,\dots,v_Q$ of incoming velocities before the Q collisions. Then
\begin{equation}\label{limitino trajectory}
\left\{\begin{array}{ll}
\bar{\xi}_{\ep}(s)=x-v(t-t_1)-v_1(t_1-t_2)\dots-v_Qt_Q&\\
\bar{\eta}_{\ep}(s)=v_Q.&
\end{array}\right.
\end{equation}
\indent We define 
\begin{equation}
\begin{split}
\label{formula6}
h_{\ep}(x,v,t)&=e^{-2\ep^{-2\alpha}|v|t}\sum_{Q\geq 0}\mu_{\ep}^{Q}\int_{0}^{t}dt_1\dots\int_{0}^{t_{Q-1}}dt_Q\\&
\int_{-\ep}^{\ep}d\rho_1\dots\int_{-\ep}^{\ep}d\rho_Q\, f_0(\bar{\xi}_{\ep}(-t),\bar{\eta}_{\ep}(-t)).
\end{split}
\end{equation}

Due to the Lipschitz continuity of $f_0$ we can assert that
\begin{equation}
\bar{h}_{\ep}(x,v,t)=h_{\ep}(x,v,t)+\ffi_2(x,v,t)
\end{equation}
where 
\begin{equation}\label{phi2}
\sup_{x,v,t\in[0,T]}|\ffi_2(x,v,t)|\leq C\ep^{1-2\alpha}T.
\end{equation}
For more details see \cite{DP}, Section 3.
As matter of facts, since we realize that $h_{\ep}$ is the solution of the following Boltzmann equation
\begin{equation}
\label{Boltzmann epsilon}
(\partial_t+v\cdot\nabla_x)h_{\ep}(x,v,t)=\text{L}_{\ep}h_{\ep}(x,v,t),
\end{equation}
where 
\begin{equation}
\text{L}_{\ep}h(v)=\mu\ep^{-2\alpha}|v|\int_{-1}^{1}\,d\rho\{h(v')-h(v)\},
\end{equation}
we have reduced the problem, thanks to Proposition 1, to the analysis of a Markov process which is an easier task.\\


\section{Proof of the main theorem}
\indent 
Let be $\eta_\ep=|\log \ep|$.
We rewrite the linear Boltzmann equation \eqref{Boltzmann epsilon} in the following way
\begin{equation}\label{BE0}
\big(\partial_t +v\cdot\nabla_x\big)h_{\ep}(x,v,t)=\eta_\ep\,\tilde{ \text{L}}_{\ep} h_{\ep}(x,v,t),
\end{equation}
where $\tilde L_\ep=L/\eta_\ep$, namely
\begin{equation}\label{def:L_ve}
\tilde{ \text{L}}_{\ep}f (v)=\mu|v|\frac{\ep^{-2\alpha}}{|\log\ep|}\int_{-1}^1d\rho\big[f(v')-f(v)\big],\qquad f\in L^1(\R^2).
\end{equation}

\indent
We will show that  for $\eta\to \infty$ we get a trivial result (Theorem \ref{th:main th}, item 1)), 
then we should look at the solution for times $\eta_\ep t$, namely in the diffusive scaling.
Denoting by 
$\h:=h_\ep(x , v, \eta_\ep t)$, where $h_\ep$ solves \eqref{BE0},
$\h$ solves 
\begin{equation}\label{BE1}
\big(\partial_t  +\eta_\ep\,v\cdot\nabla_x\big)\h=\eta_\ep^2\, \tilde{ \text{L}}_{\ep}\h.
\end{equation}
\indent
It is convenient to introduce the Cauchy problem associated to the following rescaled Landau equation:
\begin{equation}\label{rescLandau}
\left\{
 \begin{array}{l}\vspace{0.2cm}
\big(\partial_t +\eta\, v\cdot\nabla_x\big)g_{\eta}(x,v,t)=\eta^2  \LL g_\eta(x,v,t),\\ 
g_\eta(t=0)=f_0.
\end{array} \right.
\end{equation}
where $\LL=\frac{\mu}{2}\frac 1 {|v|}\Delta_{|_{S_{|v|}}}$.
We observe preliminarily that eq. \eqref{rescLandau} propagates the regularity of the derivatives with respect to the $x$ variable and, due to the presence of $\LL$, gains regularity  with respect to the transverse component of the velocity. 
Indeed, for any fixed $|v|$, denoting by $S_{|v|}$ the circle of radius $|v|$, under the hypothesis of Theorem \ref{th:main th} on $f_0$, the solution $g_\eta:\,\R^2\times S_{|v|}\to \R^{+}$ satisfies the bounds
\begin{equation}\label{bounds}
|D_{x}^{k}g_\eta(x,v)|\leq C,\quad |D_{v}^{h}g_\eta(x,v)|\leq C\quad\forall k\leq 2,\,h\geq 0,
\end{equation}
$\forall t\in (0,T]$,  where $C=C(f_0,T)$ and $D_{v}$ is the derivative with respect to the transverse component of the velocity. In particular the solutions of \eqref{rescLandau} we are considering are classical.\\
\indent Before analyzing 
the asymptotic behavior of the solution  of \eqref{rescLandau} we first need  a preliminary Lemma.

\begin{lemma}\label{th:preliminary}
Let $\langle g_\eta\rangle$ be the average of $g_\eta$ with respect to the invariant measure $\nu$,
namely 
$
\langle g_\eta\rangle :=\frac 1 {2\pi} \frac 1 {|v|}\int_{S_{|v|}}
dv\,g_\eta(x,v).
$
Under the hypothesis of Theorem \ref{th:main th} 
\begin{itemize}
\item[(1)]$\quad g_\eta-\langle g_\eta\rangle\underset{\eta\to \infty}{\longrightarrow 0} \quad \text{in}\quad L^{\infty}((0,T];L^2(\R^2\times S_{|v|})).$
\end{itemize}
Moreover, setting $t_{\eta}=\frac{1}{\eta^{\omega}}$ for $\omega>2$ then
\begin{itemize}
\item[(2)]$\quad g_\eta(t_\eta)-\langle f_0\rangle\underset{\eta\to\infty}{\longrightarrow 0} \,\text{in}\quad L^2(\R^2\times S_{|v|}),$
\end{itemize}
where $\langle f_0\rangle=\frac 1 {2\pi}\frac 1 {|v|}\int_{S_{|v|}}dv\, f_0.$
\end{lemma}
\begin{proof}
Let $R_{\eta}=g_\eta-\langle g_\eta\rangle$. We have
\begin{equation}\label{L1}
\big(\partial_t +\eta\, v\cdot\nabla_x\big)R_{\eta}(x,v,t)=\eta^2 \LL R_\eta(x,v,t)+\ffi,
\end{equation}
where
\begin{equation}\label{L2}
\begin{split}
\ffi&=-\big(\eta\, v\cdot\nabla_x\langle g_\eta\rangle+\partial_{t}\langle g_\eta\rangle\big)\\&
=\eta \Big(\frac{1}{2\pi}\frac 1 {|v|}\int_{S_{|v|}}v'\cdot\nabla_x g_{\eta}\,dv'-v\cdot\nabla_x \langle g_{\eta}\rangle\Big).
\end{split}
\end{equation}
We can estimate the last quantity by \eqref{bounds}:
$$\sup_{t\leq T}\|\ffi\|_{L^2}\leq C\eta\|\nabla_x g_{\eta}\|_{L_2}\leq C\eta.$$ 
Therefore by \eqref{L1} we have 
\begin{equation*}\begin{split}
\frac{1}{2}\frac{d}{dt}\|R_{\eta}(t)\|_{L^2}^2&=\eta^2(R_{\eta},\LL R_{\eta})+(R_{\eta},\ffi)\\&
\leq -\eta^2 \lambda\|R_{\eta}\|^2_{L^2}+\|R_{\eta}\|_{L^2}\|\ffi\|_{L^2}.
\end{split}
\end{equation*}
where $\lambda$ is the first positive eigenvalue of $\LL$. Here we used that $R_{\eta}\perp 1$ in $L^2$.
Hence
\begin{equation*}
\begin{split}
\|R_{\eta}(t)\|_{L^2}&\leq e^{-\eta^2 \lambda t}\|R_{\eta}(0)\|_{L^2}+\int_{0}^{t}\,ds\, e^{-\eta^2 \lambda (t-s)}\|\ffi(s)\|_{L^2}\\&
\leq e^{-\eta^2 \lambda t}\|R_{\eta}(0)\|_{L^2}+\frac{C}{\eta}(1-e^{-\eta^2 \lambda t}), 
\end{split}
\end{equation*}
so that (1) is proven.\\
\indent To prove (2) observe that, thanks to the fact $\LL$ is negative, we have 
\begin{equation*}
\begin{split}
\frac{1}{2}\frac{d}{dt}\|g_{\eta}(t)-f_0\|_{L^2}^2
&
\leq -\eta(g_{\eta}-f_0, v\cdot\nabla_x f_0)
+\eta^2 (g_{\eta}-f_0,\LL f_0)\\
& \leq   \|g_{\eta}-f_0\|_{L^2}\big(\eta |v|\,\|\nabla_x f_0\| +\eta^2 \|\LL f_0\|\big).
\end{split}
\end{equation*}
\noindent
Therefore
\begin{equation}\label{L3}
\|g_{\eta}(t_\eta)-f_0\|_{L^2}\leq \frac 1 {\eta^{\omega}}\big(\eta |v|\,\|\nabla_x f_0\| +\eta^2 \|\LL f_0\|\big),
\end{equation}
which vanishes as $\eta\to\infty$.
Finally, recalling that $\langle f_0\rangle=\frac 1 {2\pi}\frac 1 {|v|}\rho_0$, we have 
\begin{equation*}
\begin{split}
\|g_{\eta}(t_\eta)-\langle f_0\rangle\|_{L^2}&
\leq \sup_{t\in (0,T]} \|g_{\eta}-\langle g_\eta\rangle\|_{L^2}+\|\langle g_\eta(t_\eta)\rangle-\langle f_0\rangle\|_{L^2}\\&
\leq \sup_{t\in (0,T]}\|g_{\eta}-\langle g_\eta\rangle\|_{L^2}+c\|g_{\eta}(t_\eta)-f_0\|_{L^2}.
\end{split}
\end{equation*}
By \eqref{L3} and (1) we conclude the proof.
\end{proof}

\begin{lemma}\label{th:diff}
Let $g_\eta$ be the solution of \eqref{rescLandau}.
Under the hypothesis of Theorem \ref{th:main th} for the initial datum $f_0$,  for $\eta\to \infty$ 
$g_\eta$ converges to the solution of the diffusion equation
\begin{equation}
\left\{
 \begin{array}{l}\vspace{0.2cm}
\partial_t\varrho=D\Delta\varrho\\ 
\varrho(x,0)=\langle f_0\rangle, 
\end{array} \right.
\end{equation}
where $\langle f_0\rangle=\frac 1 {2\pi}\frac 1 {|v|}\int_{S_{|v|}}dv\, f_0$ and 
\begin{equation}
D=\frac 1 {\mu} |v|\int_{S_{|v|}} v\cdot \Delta_{|_{S_{|v|}}}^{-1}v\,dv.
\end{equation}
Convergence is in  $L^{\infty}([0,T];L^2(\R^2 \times S_{|v|}))$.
\end{lemma}

\begin{proof} The proof of the above Lemma is rather straightforward (see e.g. \cite{EP}).\\
\indent
Suppose for the moment that the initial datum depends only on the position variables, namely the initial datum has the form of a local equilibrium.
We assume that $g_\eta$ has the following form
$$g_\eta(x,v,t)=g^{(0)}(x,t)+\frac 1 {\eta}g^{(1)}(x,v,t)+\frac 1 {\eta^2}g^{(2)}(x,v,t)
+\frac 1 {\eta}R_{\eta},$$
where $g^{(i)}$, $i=0,1,2$ are the first three coefficient of a Hilbert expansion in $\eta$,
and $R_{\eta}$ is the reminder.
Comparing terms of the same order in $\eta$ we obtain   the following equations:
\begin{equation*}\begin{split}
&(i)  \,v\cdot\nabla_x g^{(0)}=\frac{\mu} 2 \frac 1 {|v|}\,\Delta_{|_{S_{|v|}}} g^{(1)}\\
&(ii)  \,\partial_t\, g^{(0)}+v\cdot\nabla_x g^{(1)}=\frac{\mu} 2 \frac 1 {|v|}\,\Delta_{|_{S_{|v|}}}
 g^{(2)}\\
&(iii)  \,\big(\partial_t  +\eta\,v\cdot\nabla_x\big)R_{\eta}=\eta^2 \frac{\mu} 2 \frac 1 {|v|}\,\Delta_{|_{S_{|v|}}} R_{\eta}
-A_\eta(t),
\end{split}\end{equation*}
with $A_\eta(t)=A_\eta(x,v, t)=\partial_t g^{(1)}+\frac 1 {\eta} g^{(2)}
+v\cdot\nabla_x g^{(2)}$.
\indent
Since $v\cdot \nabla_x g^{(0)}$ is an odd function of $v$,  the integral with respect to $v$ of the left hand side of (i) vanishes. Then we can invert the operator $\Delta_{|_{S_{|v|}}}$ and set 
$g^{(1)}=\frac 2 {\mu} |v|\Delta_{|_{S_{|v|}}}^{-1}v\cdot\nabla_x g^{(0)}$, where $g^{(1)}$ is an odd function of the velocity.
Now we integrate the second equation with respect to the velocity. By observing that 
$\int_{S_{|v|}} dv\,\Delta_{|_{S_{|v|}}}g^{(2)}=0$, since $dv_{|_{S_{|v|}}}$ is proportional the invariant measure, we obtain
\begin{equation*}
\partial_t \, g^{(0)}+\frac 2 {\mu} |v|\,\int_{S_{|v|}}dv\, v\cdot \nabla_x\big(
\Delta_{|_{S_{|v|}}}^{-1}v\cdot\nabla_x g^{(0)}\big)=0.
\end{equation*}
We define the $2\times 2$ matrix $D$ as $D_{ij}=-\frac 2 {\mu} |v|\int_{S_{|v|}} v_i \Delta_{|_{S_{|v|}}}^{-1}v_j$
and we observe that $D_{ij}=0$ for $i\neq j$ and
$D_{11}=D_{22}=D$, where
$$
D=\frac 1 {\mu}|v|\,\int_{S_{|v|}} dv\,v\cdot\big(-\Delta_{|_{S_{|v|}}}^{-1}  \big)v.
$$
 
Therefore
\begin{equation*}
\partial_t \, g^{(0)}-D\Delta_x g^{(0)}=0,
\end{equation*}
where $g^{(0)}$ satisfies the initial condition $g^{(0)}(x,0)=g(t=0)$. Moreover, the $L^2$-norm of 
$g^{(1)}$ is bounded. If we show that also the $L^2$-norm of $g^{(2)}$ and $R_\eta$ are bounded, we deduce that $g_\eta$ converges to $g^{(0)}$ for $\eta\to\infty$.\\
\indent
From equation $(ii)$ and the diffusion equation for $g^{(0)}$ we derive that the integral with respect 
to $v$ of the left hand side of $(ii)$ vanishes. Therefore we can invert the operator $\Delta_{|_{S_{|v|}}}$ and obtain
\begin{equation*}\begin{split}
g^{(2)}= & \frac 2 {\mu} |v|\,\Delta_{|_{S_{|v|}}}^{-1}\big(\partial_t g^{(0)} +v\cdot \nabla_x(v\cdot\nabla_x)g^{(0)} \big)\\
=& \frac 2 {\mu} |v|\sum_{i,j}\partial_{x_i}\partial_{x_j}g^{(0)} \,\Delta_{|_{S_{|v|}}}^{-1}\big[v_i v_j+D_{ij} \big].
\end{split}\end{equation*}
Therefore the $L_2$-norm of $g^{(2)}$ is bounded.

We derive from equation $(iii)$
\begin{equation*}
\frac 1 2 \partial_t \|R_\eta\|^2 =-\eta^2\, \big(R_\eta,\,-\Delta_{|_{S_{|v|}}} R_\eta \big)-\big(R_\eta,\, A_\eta (t)  \big),
\end{equation*}
where $(\cdot,\,\cdot)$ denotes the scalar product in $L^2$.
Using positivity of $-\Delta_{|_{S_{|v|}}}$ and Cauchy-Schwartz we deduce 
$
\partial_t \|R_\eta\|\leq \|A_\eta\|.
$
Recall the explicit expression for $A_\eta$, namely
$A_\eta=\partial_t g^{(1)}+\frac 1 {\eta} g^{(2)}
+v\cdot\nabla_x g^{(2)}$.
 By direct computation
\begin{equation*}\begin{split}
\partial_t g^{(1)} &= \frac 2 {\mu}|v|\sum_{i,j, k}\partial_{x_i}\partial_{x_j}\partial_{x_k}g^{(0)}\,
 \Big[\frac 2 {\mu}|v|v_i \Delta_{|_{S_{|v|}}}^{-1} v_j -(v_iv_j+D_{i,j}) \Big] \Delta_{|_{S_{|v|}}}^{-1}v_k,
\end{split}\end{equation*}
from which we deduce that the $L^2$-norm of $\partial_t g^{(1)}$ is bounded. Similarly, one can easily show that
the $L^2$-norm of $v\cdot\nabla_x g^{(2)}$ is bounded, and then
$\|A_\eta\|$ is uniformly bounded in $[0,T]$ and $\|R_\eta \|\leq C T$.

To complete the proof we consider more general initial data $f_0$ depending also on the velocity variable.
Let $A:=L-\eta v\cdot\nabla_x$. 
We compare $g_\eta$ with $\bar g_\eta$, the solution  \eqref{rescLandau} with initial datum $\langle f_0\rangle$.
By  the same argument as in  Lemma \ref{th:preliminary}, item (2), we have that $\forall t\geq t_\eta$
\begin{equation*}
\|\bar g_{\eta}(t-t_{\eta})-\bar g_{\eta}(t)\|_{L^2}\leq \frac{C}{\eta^{\omega-2}},
\end{equation*}
where $C$ depends on the $L^2$-norm of $\langle f_0\rangle$ and $\nabla\langle f_0\rangle$.
Since  $g_{\eta}(t)=e^{A t}f_0=e^{A(t-t_\eta)}g_\eta(t_\eta)$ and $\bar g_{\eta}(t-t_\eta)=e^{A(t-t_\eta)}\langle f_0\rangle$ we derive
\begin{equation*}
\|g_{\eta}(t)-\bar g_{\eta}(t-t_{\omega})\|_{L^2}\leq C\|g(t_{\eta})-\langle f_0\rangle\|_{L^2}.
\end{equation*}
Thus, by Lemma \ref{th:preliminary}, item (2), we obtain that $g_{\eta}(t)$ and $\bar g_{\eta}(t)$ have the same asymptotics and this concludes the proof of Lemma \eqref{th:diff}.
\end{proof}

\begin{proposition}\label{th:prop 2}
Let $f_0$ be an initial datum for
$\h$ solution of \eqref{BE1}. Under the hypothesis of Theorem \ref{th:main th}
$\h$
converges to 
$\varrho$ as $\ep\to 0$,
where $\varrho:\R^2\times[0,T]\to\R_+$ is the solution of the diffusion equation
\begin{equation}
\left\{
 \begin{array}{l}\vspace{0.2cm}
\partial_t\varrho=D\Delta\varrho\\ 
\varrho(x,0)=\langle f_0\rangle, 
\end{array} \right.,
\end{equation}
with $\langle f_0\rangle=\frac 1 {2\pi}\frac 1 {|v|}\int_{S_{|v|}}dv\, f_0$.
The diffusion coefficient is $D$ given by the Green-Kubo formula. Convergence is in $L^2(\R^2 \times S_{|v|})$
uniformly in $t\in(0,T]$. 
\end{proposition}

\begin{proof}
Let $g_{\eta_\ep}$ be solution of \eqref{rescLandau} with $\eta=\eta_\ep:=|\log\ep|$
and initial condition $f_0$. 
We look at the evolution of $\h-g_{\eta_\ep}$, 
namely
\begin{equation*}
\big(\partial_t+\eta_\ep\, v\cdot\nabla_x \big)\big(\h-g_{\eta_\ep})
=\eta_\ep ^2\Big(\tilde{\text{L}}_\ep\h-\LL g_{\eta_\ep} \Big),
\end{equation*}
where $\LL:=\frac{\mu}{2}\frac 1 {|v|}\Delta_{|_{S_{|v|}}}$.
Then we obtain 
\begin{equation*}\begin{split}
\frac 1 2 \,\partial_t \|\h-g_{\eta_\ep}\|^2= &
-\eta_\ep^2 \big(\h-g_{\eta_\ep},\,
- \tilde{\text{L}}_\ep\big[\h-g_{\eta_\ep} \big]  \big)\\
&+\eta_\ep^2
\big(\h-g_{\eta_\ep} ,\, \big[\tilde{\text{L}}_\ep-\LL \big] g_{\eta_\ep} \big),
\end{split}
\end{equation*}
from which, using positivity of $-\tilde{\text{L}}_\ep$ and Cauchy-Schwartz,
\begin{equation*}
\frac 1 2 \,\partial_t \|\h-g_{\eta_\ep}\|\leq \eta_\ep^2\big\|\big(\tilde{\text{L}}_\ep-\LL \big)
g_{\eta_\ep}\big\|.
\end{equation*}
Recalling that
\begin{equation*}
\tilde{\text{L}}_\ep g_{\eta_\ep}=\mu |v|\frac{\ep^{-2\alpha}}{|\log\ep|}\int_{-1}^1 d\rho\,
\big[g_{\eta_\ep}(x, v', t)-g_{\eta_\ep}(x,v,t)  \big],
\end{equation*}
we set
\begin{equation*}\begin{split}
&g_{\eta_\ep}(v')-g_{\eta_\ep}(v)\\
&=  (v'-v)\cdot\nabla_{|_{S_{|v|}}}g_{\eta_\ep}(v)\\
&\quad+\frac 1 2  (v'-v)(v'-v)\otimes (v'-v)\nabla_{|_{S_{|v|}}}\nabla_{|_{S_{|v|}}}g_{\eta_\ep}(v)\\
&\quad+\frac 1 6 (v'-v)\otimes (v'-v)\otimes (v'-v)
\nabla_{|_{S_{|v|}}}\nabla_{|_{S_{|v|}}}\nabla_{|_{S_{|v|}}}g_{\eta_\ep}(v)
+R_{\eta_\ep},\\
\end{split}\end{equation*}
with $R_{\eta_\ep}=\mathcal O (|v-v'|^4)$.
Integrating with respect to $v$ and using symmetry arguments we obtain 
\begin{equation*}\begin{split}
\tilde{\text{L}}_\ep g_{\eta_\ep}=\mu |v|\frac{\ep^{-2\alpha}}{|\log\ep|}\big\{
\frac 1 2 \Delta_{|_{S_{|v|}}}g_{\eta_\ep}\int_{-1}^1 d\rho\,|v'-v|^2
+\int_{-1}^1 d\rho\,R_{\eta_\ep}\big\}.
\end{split}
\end{equation*}
Observe that $|v'-v|^2=4\sin^2\frac{\theta(\rho)}{2}$, then  
by direct computation (see Appendix)
\begin{equation*}
\lim_{\ep\to 0}\,
\frac{\ep^{-2\alpha}}{|\log\ep|}
\int_{-1}^1 d\rho\,|v'-v|^2=2\frac {\alpha}{|v|^2}
\end{equation*}
and
$$
\frac{\ep^{-2\alpha}}{|\log\ep|}\int_{-1}^1d\rho\,
|v-v'|^4=\ep^\alpha|\log\ep|^\beta,\qquad -1<\beta
<\frac  5 2 \alpha -1.$$
Therefore $\big\|\big(\tilde{\text{L}}_\ep-\LL \big)
g_{\eta_\ep}\big\|\leq \ep^\alpha|\log\ep|^\beta\,\|\Delta^2_{|_{S_{|v|}}}g_{\eta_\ep}\|\leq \ep^\alpha|\log\ep|^\beta\, C$,
which vanishes for $\ep\to 0$.

\end{proof}

In order to complete the proof of the item 3) of Theorem \ref{th:main th}, we need to show that   $f_{\ep}(\eta_{\ep} t)$ converges to $\tilde h_\ep(t)$
in $L^2(\R^2\times \R^2)$, for every $t\in [0,T]$. By Proposition \ref{th:prop} and Remark \ref{th:remarkprop} we have that $\bar f_{\ep}(\eta_\ep t)$ defined in \eqref{formula4} converges to 
$\bar h_{\ep}(\eta_\ep t)$, \eqref{formula5}, in $L^1(\R^2\times \R^2)$, for every $t\in [0,T]$. Moreover, using \eqref{phi2}
and the fact that the initial datum has compact support, we have that $\bar h_\ep(\eta_\ep t)$ 
converges to
$\tilde h_\ep(t)$ in $L^1(\R^2\times \R^2)$, for every $t\in [0,T]$. Under hypothesis of Theorem \ref{th:main th}, convergence in $L^1$ norm implies convergence in $L^2$.
 Since $\bar f_{\ep}\leq f_\ep$ and using the fact that  at $t=0$ the equality holds and the  linear Boltzmann equation \ref{BE1} preserves the total mass, then also  
$ f_{\ep}(\eta_\ep t)$ converges to $\tilde h_\ep(t)$ in 
$L^2(\R^2\times S_{|v|})$, for every $t\in [0,T]$. 
 
Now we go back to equation \eqref{BE0}. Using the same strategy of  the proof of Proposition \ref{th:prop 2} we can replace $\tilde L_\ep$ with $\mathcal L$,
and we denote  $\tilde g_\eta $ the solution of
\begin{equation*}
\big(\partial_t+v\cdot\nabla_x\big)\tilde g_\eta=\eta \,\mathcal L \tilde g_\eta,
\end{equation*}
 with initial datum $f_0$. By the same arguments as in Lemma \ref{th:preliminary}, item (i), one can prove that for $\eta\to\infty$ $\tilde g_\eta\to\langle \tilde g_\eta\rangle$ and
$\nabla_x \tilde g_\eta \to \nabla_x\langle \tilde g_\eta\rangle$. We observe that
$$
\partial_t \langle \tilde g_\eta\rangle +\nabla_x \int dv\, \big(\tilde g_\eta-\langle\tilde g_\eta\rangle \big) v=0,
$$
therefore $\langle \tilde g_\eta\rangle$ converges to $\langle \tilde f_0\rangle$ as $\eta\to\infty$, which concludes the proof of item 1).

Proof of item 2) is included in the proof of Proposition \ref{th:prop 2}.

%
%


\section{The control of the pathological sets}
\indent
In this section we prove Proposition \ref{th:prop}.\\
Clearly
\begin{equation}
\label{error}
1-\chi_1(1-\chi_{ov})(1-\chi_{rec})(1-\chi_{int}) \leq (1-\chi_1)+\chi_{ov}+\chi_{rec}+\chi_{int}
\end{equation}
and we estimate separately all the events in the right hand side of  \eqref{error}.\\
\indent We denote by $\xi_{\ep}(s), \eta_{\ep}(s)$ the backward Markov process defined, for $s\in(-t,0)$, in Section 2 and we set
\begin{equation}
\begin{split}
\label{Markov}
\EE _{x,v}^{M}(u)&=e^{-2|v|\ep^{-2\alpha}t} \sum_{Q\geq 0}(2|v|\mu_{\ep}t)^Q \int_{0}^{t}dt_1\dots\int_{0}^{t_{Q-1}}dt_Q\\&
\int_{-\ep}^{\ep}d\rho_1\dots\int_{-\ep}^{\ep}d\rho_Q \,u(\xi_{\ep},\eta_{\ep}),
\end{split}
\end{equation}
for any measurable function $u$ of the process $(\xi_{\ep},\eta_{\ep})$. We have
\begin{equation}
\begin{split}
\label{error1}
&\EE _{x,v}^{M}((1-\chi_1)f_0(\xi_{\ep}(-t),\eta_{\ep}(-t)))\\&
\leq \frac{2\ep}{|v|}e^{-2|v|\ep^{-2\alpha}t} \| f_0\|_{L^{\infty}} \sum_{Q>0}\frac{(2|v|\ep^{-2\alpha})^Q}{(Q-1)!} t^{Q-1}\\&
\leq 4 \| f_0\|_{L^{\infty}}\ep^{1-2\alpha}t\leq C\ep^{\gamma}t,
\end{split}
\end{equation}
for $\gamma>0$, $\alpha<1/2$ and $\ep$ sufficiently small.\\ Here and in the sequel $t$ is allowed to behave as $c|\log(\ep)|$.\\
\indent Estimate \eqref{error1} is obvious. Indeed if $\chi_1=0$ the first or the last collision must satisfy either $|t-t_1|\leq 2\ep/|v|$ or $t_Q\leq 2\ep/|v|$. Hence \eqref{error1} follows easily.\\
A similar argument can be used to estimate $\chi_{ov}$. Indeed if $\chi_{ov}=1$ it must be $t_{i}-t_{i+1}\leq2\ep/|v|$ for some $i=1,\dots,(Q-1)$. Therefore proceeding as before
\begin{equation}
\begin{split}
\label{Markov error ov}
&\EE _{x,v}^{M}(\chi_{ov}f_0(\xi_{\ep}(-t),\eta_{\ep}(-t)))\\&
\leq \frac{2\ep}{|v|}e^{-2|v|\ep^{-2\alpha}t} \| f_0\|_{L^{\infty}} \sum_{Q> 1}(Q-1)\frac{(2|v|\ep^{-2\alpha})^Q}{(Q-1)!} t^{Q-1}\\&
\leq \frac{2\ep}{|v|} \| f_0\|_{L^{\infty}}t(2|v|\ep^{-2\alpha})^2\leq C|v|\ep^{\gamma}t,
\end{split}
\end{equation}
for some $\gamma>0$, $\alpha<1/4$ and $\ep$ sufficiently small.\\
\indent Next we pass to the control of the recollision event. We proceed similarly as in \cite{DP} and in \cite{DR}. Let $t_i$ the first time the light particle hits the i-th scattering, $\eta_i^{-}$ the incoming velocity, $\eta_i^{+}$ the outgoing velocity and $t_i^{+}$ the exit time. Moreover we fix the axis in such a way that $\eta_i^{+}$ is parallel to the $x$ axis (see figure \ref{fig:7}). We have
\begin{equation} 
\chi_{rec}\leq \sum_{i=1}^{Q}\sum_{j>1}\chi_{rec}^{i,j},
\end{equation}
where  $\chi_{rec}^{i,j}=1$ if and only if $b_i$ (constructed via the sequence $t_1,\rho_1,\dots,t_i,\rho_i$) is recollided in the time interval $(t_j^{+},t_{j+1})$.\\
\begin{figure}
\centering
\includegraphics{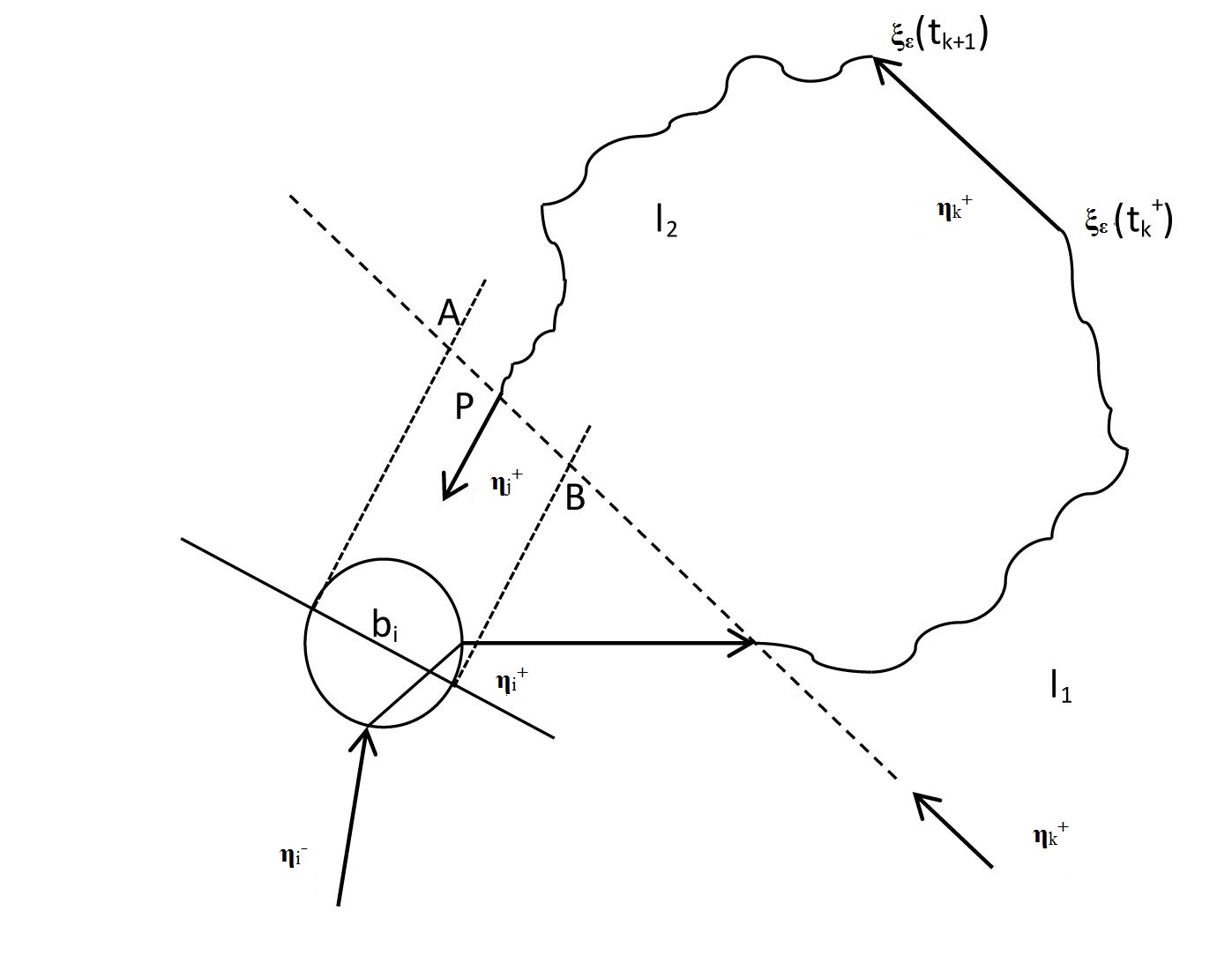}
\caption{}\label{fig:7}
\end{figure}

Note that, since $|\theta_i|\leq C\ep^{\alpha}$, where $\theta_i$ is the i-th scattering angle, in order to have a recollision it must be an intermediate velocity $\eta_k$, $k=i+1,\dots,j-1$ such that
\begin{equation}
\label{vel.cdz}
|\eta_{k}^{+}\cdot \eta_{j}^{+}|\leq C \ep^{\alpha}|v|^2,
\end{equation} 
namely $\eta_k^{+}$ is almost orthogonal to $\eta_j^{+}$ (see the figure). Then
\begin{equation} 
\chi_{rec}\leq \sum_{i=1}^{Q}\sum_{j=1}^{Q}\sum_{k=i+1}^{j-1}\chi_{rec}^{i,j,k},
\end{equation}
where $\chi_{rec}^{i,j,k}=1$ if and only if $\chi_{rec}^{i,j}=1$ and \eqref{vel.cdz} is fulfilled.\\
\indent Fix now all the parameters $\rho_1,\dots,\rho_{Q}$, $t_1,\dots,t_Q$ but $t_{k+1}$ and perform such a time integration. The two branches of the trajectory $l_1,l_2$ are rigid so that, if the recollision happen the time integration with respect to $t_{k+1}$ is restricted to a time interval proportional to $AB$. More precisely it is bounded by $$\frac{2\ep}{|v|\cos C\ep^{\alpha}}\leq \frac{2\ep}{|v|}.$$ 		
Performing all the other integrations and summing over $i,j,k$ we obtain
\begin{equation}
\begin{split}
\label{Markov error rec}
&\EE _{x,v}^{M}(\chi_{rec}f_0(\xi_{\ep}(-t),\eta_{\ep}(-t)))\\&
\leq \frac{4\ep}{|v|}e^{-2|v|\ep^{-2\alpha}t} \| f_0\|_{L^{\infty}} \sum_{Q\geq 3}(Q-1)(Q-2)(Q-3)\frac{(2|v|\ep^{-2\alpha})^Q}{(Q-1)!} t^{Q-1}\\&
\leq C|v|^3 t^3\ep^{1-8\alpha}\leq C|v|^3\ep^{\gamma}t^3,
\end{split}
\end{equation}
for some $\gamma>0$, $\alpha<1/8$ and $\ep$ sufficiently small.\\
\indent We finally estimate the event $\chi_{int}$. To do this we fix a sequence of parameters $\rho_1,\dots,\rho_{Q}$, $t_1,\dots,t_Q$. For instance consider the case in figure \ref{fig:8} in which we exhibit an unphysical trajectory.

\begin{figure}
\centering
\includegraphics[scale=0.8]{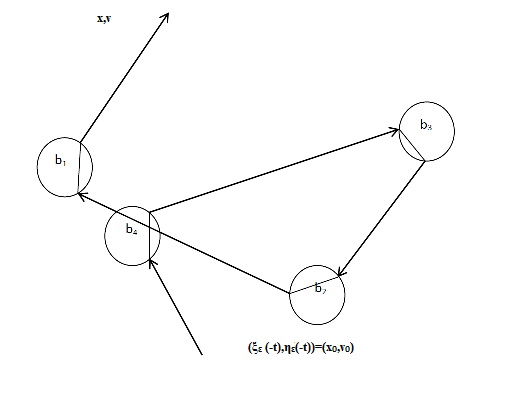}
\caption{}\label{fig:8}
\end{figure}

Consider the integral
\begin{equation}
I=\int_{B(0,M)} f_0(\xi(-t),\eta(-t))\chi_{int}\,dxdv.
\end{equation}
Here $\chi_{int}=1$ for those values of $x,v$ for which an interference takes place and
\begin{equation*}
\noindent B(0,M):=\{(x,v)\in \R^2\times\R^2;\;\text{s.t.}\; |x|^2+|v|^2<M\}. 
\end{equation*}
By the Liouville Theorem we can integrate over the variables $\big(\xi_{\ep}(-t),\eta_{\ep}(-t)\big)=(x_0,v_0)$ as independent variables
\begin{equation}
I=\int_{B_t(0,M)} f_0(x_0,v_0)\chi_{rec}\,dx_0dv_0,
\end{equation}
where 
\begin{equation*}
\noindent B_t(0,M)=\left\{\big(\xi_{\ep}(-t),\eta_{\ep}(-t)\big)\;\text{s.t.}\; (x,v)\in B(0,M)\right\}.
\end{equation*}
Note that $$\chi_{int}(x,v)=\chi_{rec}(x_0,v_0),$$ since a backward interference is a forward recollision.
Clearly
\begin{equation}
B_t(0,M)\subset B(0,M(1+t))
\end{equation}
where $B(0,M(1+t))$ is the ball of radius $M(1+t)$ in $\R^4$.\\
Thus 
\begin{equation}
\label{I 3}
I\leq\int_{B(0,M(1+t))}f_0(x_0,v_0)\chi_{rec}\,dx_0dv_0.
\end{equation}
Therefore, by using estimate \eqref{Markov error rec} and \eqref{I 3}
\begin{equation}
\begin{split}
\label{interferences}
&\int_{B(0,M)}\EE _{x,v}^{M}(\chi_{int}f_0(\xi_{\ep}(-t),\eta_{\ep}(-t)))\,dxdv\\&\leq C\ep^{\gamma}M^7(1+t)^7.
\end{split}
\end{equation}
This concludes the proof of Proposition \ref{th:prop}.

\section{Concluding Remarks}
\indent The diffusive limit analyzed in the present paper is suggested by the divergence of $B$ for the particular choice of the potential we are considering. However the same techniques could work in presence of a smooth, radial, short-range potential $\phi$.
\begin{theorem}\label{th:main th2}
Under the same hypothesis of Theorem \ref{th:main th}, assume $\phi\in C^2([0,1])$. Scale the variables, the density and the potential according to 
\begin{equation}\label{scaling 2}
\left\{
 \begin{array}{l}\vspace{0.2cm}
x\to\ep x\\ 
t\to\ep^{\lambda}\ep t\\
\mu_{\ep}=\ep^{-(2\alpha+\lambda+1)}\mu\\
\phi\to\ep^{\alpha}\phi.
\end{array} \right.
\end{equation}
Then, for $t>0$ and $\ep\to 0$, there exists $\lambda_0=\lambda(\alpha)$ s.t. for $\lambda<\lambda(\alpha)$
\begin{equation*}
f_{\ep}(x,v,t)\to \rho(x,t)
\end{equation*}
solution of the heat equation
\begin{equation}
\left\{
 \begin{array}{l}\vspace{0.2cm}
\partial_t\varrho=D\Delta\varrho\\ 
\varrho(x,0)=\langle f_0\rangle, 
\end{array} \right.
\end{equation}
with $D$ given by the Green-Kubo formula
\begin{equation}\label{GK2}
D=\frac{1}{\mu}|v|\int_{S_{|v|}}{v\cdot\big(-\Delta_{|_{S_{|v|}}}^{-1}\big)v\,dv}.
\end{equation}
The convergence is in $L^2(\R^2\times\R^2)$. 
\end{theorem}

The significance and the proof of the above theorem is clear. The kinetic regime describes the system for kinetic times $O(1)$. One can go further to diffusive times provided that $\lambda$ is not too large. Indeed the distribution function $f_{\ep}$ ``almost'' solves 
\begin{equation}
\begin{split}
\big(\ep^{\lambda}\partial_t+v\cdot\nabla_x\big)f_{\ep}&\approx\ep^{-2\alpha-\lambda}\text{L}_{\ep}f_{\ep}\\&
\approx \ep^{-\lambda}c\Delta_{|_{S_{|v|}}}f_{\ep},
\end{split}
\end{equation}
for which the arguments of Section 3 do apply. 
In other words there is a scale of time for which the system diffuses. However such times are not so large to prevent the Markov property. 
Obviously the diffusion coefficient is computed in terms of the limiting Markov process.
We can give an estimate , certainly not optimal, of the coefficient $\lambda$ appearing in \eqref{scaling 2}. Estimating recollisions and interferences as in Section 4, setting $\gamma=1-8(\alpha+\frac{\lambda}{2})$, the condition on $\lambda$ is (see \eqref{interferences})
$$\gamma-7\lambda>0\quad\text{i.e.}\quad \lambda<\frac{1-8\alpha}{11}.$$
\indent Although the scaling we are considering in Theorem \eqref{th:main th2} is quite particular, the aim is the same as in \cite{LE} where the same problem has been approached for the weak-coupling limit ($\alpha=\frac{1}{2}$) of a quantum system.\\
\indent Recently we were aware of a result concerning the diffusion limit of a test particle of a hard-core system at thermal equilibrium \cite{BGS-R}.
Also in this case the quantitative control of the pathological trajectories allows to reach larger times in which a diffusive regime is outlined.

\vspace{10mm}
\indent\textbf{Acknowledgments.}\\ We are indebted to S. Simonella and H. Spohn for illuminating discussions.
\vspace{10mm}
\appendix
\section{Appendix (on the scattering problem associated to a circular potential barrier)} \label{scatt}
\setcounter{equation}{0}    
\def\theequation{A.\arabic{equation}}

\noindent The potential energy for a finite potential barrier is given by
\begin{equation}
\label{pot}
 \phi( r) =
  \begin{cases}
   \phi_0 & \text{if } r \leq 1 \\
   0       & \text{if } r >1
  \end{cases}
\end{equation}
The light particle, of unitary mass, moves in a straight line with energy  $E=\frac{1}{2}v^2>\phi_0$. Let $\rho$ be the impact parameter. 
For small impact parameters the particle will pass through the barrier, for large ones the particle will be reflected.  
Inside the barrier the velocity is a constant $v=\bar v$ ($\bar v<v$).
The complete trajectory of the light particle which passes through the barrier consists of three straight lines and is symmetrical about a radial line perpendicular to the interior path.

Let $\alpha$ be the angle of incidence (the inside angle between the trajectory and a radial line to the point of contact with tha barrier at $r=1$) and $\beta$ the angle of refraction (the corresponding external angle). We assume that the radius of the circle is $r=1$. 
According to the geometry of the problem $\alpha$ and $\beta$ are such that $$\sin\beta=\frac{v}{\bar v}\sin\alpha$$
where $\sin\alpha=\rho$.

The angle of deflection is $\theta=2(\beta-\alpha).$
Thanks to the energy and angular momentum conservation the expression for the refractive index becomes
\begin{equation}
\label{refractive index 1 }
n=\frac{\sin\alpha}{\sin\beta}=\frac{\bar v}{v}=\sqrt{1-\frac{2\phi_0}{v^2}}
\end{equation}
and so we have a scattering angle defined in the following way: 
\begin{equation}
\label{scattering angle}
\theta(\rho)=
\begin{cases}
2\left(\arcsin\left(\frac{\rho}{n}\right)-\arcsin(\rho)\right)  & \text{if } \rho \leq n \\
2\arccos(\rho)       & \text{if } \rho>n.
  \end{cases}
\end{equation}
In the first case the particle passes through the barrier (for $\rho\leq n$), and in the second one the particle is reflected (for $\rho> n$). The maximum scattering angle $\theta_{\max}=2\arccos(n)$ is the angle at which the particle scatters tangentially to the barrier.
The differential scattering cross section 
$$\Psi(\theta)=\left|\frac{\partial\rho}{\partial\theta}\right|$$
is then:
\begin{equation}\displaystyle
\label{cross section}
\Psi(\theta)=
\begin{cases}
\frac{n[\cos(\theta/2)-n][1-n\cos(\theta/2)]}{(1+n^2-2n\cos(\theta/2))^{3/2}}  & \text{if }\theta\leq 2\arccos(n) \\
[1-\cos^2(\theta/2)]^{1/2}       & \text{otherwise } 
  \end{cases}
\end{equation}

Scaling now the potential as $\phi( r) \to \ep^{\alpha}\phi( r)$, the previous formulas still hold. Thus, according to this scaling, the refractive index becomes,  
\begin{equation}
\label{refractive index scaled }
n_{\ep}=\sqrt{1-\frac{2\ep^{\alpha}\phi_0}{v^2}}
\end{equation} 
to replace into \eqref{scattering angle}, and \eqref{cross section}.

\section{Appendix (on the diffusion coefficient)} \label{diff}
\setcounter{equation}{0}    
\def\theequation{B.\arabic{equation}}

\noindent
In this section we show that the diffusion coefficient is divergent for the circular potential barrier \eqref{pot}. 
At this level we assume that $\phi_0=1$ to simplify the following expressions.

We need to compute
\begin{equation}
\label{diffusion coefficient scaled}
\tilde{B}:=\lim_{\ep\to 0}{\frac{\mu\ep^{-2\alpha}}{2}|v|\int_{-1}^{1}{\theta^2(\rho)\,d\rho}}.
\end{equation}
Thanks to the symmetry for the scattering problem
\begin{equation}
\label{B1}
\ep^{-2\alpha}\int_{-1}^{1}{\theta^2(\rho)\,d\rho}=2\ep^{-2\alpha}\int_{0}^{1}{\theta^2(\rho)\,d\rho}.
\end{equation}
According to \eqref{scattering angle}:
\begin{equation}
\label{B2}
2\ep^{-2\alpha}\int_{0}^{1}{\theta^2(\rho)\,d\rho}=2\ep^{-2\alpha}\left(\int_{0}^{n_{\ep}}{\theta^2(\rho)\,d\rho}+\int_{n_{\ep}}^{1}{\theta^2(\rho)\,d\rho}\right)
\end{equation}

Our aim is to perform a Taylor expansion of the first branch of $\theta(\rho)$ for $\rho\geq 0$, \\$\rho/n_{\ep}<(1-\delta)$, with $\delta>0$. We have
$$\arcsin\left(\rho/n_{\ep}\right)=\arcsin(\rho)+\frac{1}{\sqrt{1-\rho^2}}\left(\frac{\rho}{n_{\ep}}-\rho\right)+R_1\left(\rho/n_{\ep}\right),$$
where 
\begin{equation}
R_1\left(\rho/n_{\ep}\right)=\frac{\bar{\rho}}{2(1-\bar{\rho}^2)^{\frac{3}{2}}}\left(\frac{\rho}{n_{\ep}}-\rho\right)^2\quad \rho<\bar{\rho}<\frac{\rho}{n_{\ep}}.
\end{equation}
Then, looking at the first integral in the r.h.s of \eqref{B2}, we have to split it as
\begin{equation*}
\label{B2*}
\ep^{-2\alpha}\int_{0}^{n_{\ep}}{\theta^2(\rho)\,d\rho}=\underbrace{\ep^{-2\alpha}\int_{0}^{n_{\ep}(1-\delta)}{\theta^2(\rho)\,d\rho}}_{A}+\underbrace{\ep^{-2\alpha}\int_{n_{\ep}(1-\delta)}^{n_{\ep}}{\theta^2(\rho)\,d\rho}}_{B}.
\end{equation*}
Thus
\begin{equation}
\label{B3}
\begin{split}\displaystyle
&A=\ep^{-2\alpha}\int_{0}^{n_{\ep}(1-\delta)}{[2(\arcsin\left(\rho/n_{\ep}\right)-\arcsin(\rho))]^2\,d\rho}\\&
\leq4\ep^{-2\alpha}\left[\int_{0}^{n_{\ep}(1-\delta)}{\frac{(1-n_{\ep})^2}{n_{\ep}^2}\frac{\rho^2}{(1-\rho^2)}\,d\rho}+\int_{0}^{n_{\ep}(1-\delta)}{R_1\left(\rho/n_{\ep}\right)^2\,d\rho}\right]+\\&4\ep^{-2\alpha}\left[\left(\int_{0}^{n_{\ep}(1-\delta)}{R_1\left(\rho/n_{\ep}\right)^2}\right)^{\frac{1}{2}}\left(\int_{0}^{n_{\ep}(1-\delta)}{\frac{(1-n_{\ep})^2}{n_{\ep}^2}\frac{\rho^2}{(1-\rho^2)}\,d\rho}\right)^{\frac{1}{2}}\right].
\end{split}
\end{equation}

It is sufficient to compute the first two integrals . Let $A_1$ and $A_2$ be the first and the second integrals respectively. We have
\begin{equation}
\label{B4}
\begin{split}
&A_1=\ep^{-2\alpha}\frac{(1-n_{\ep})^2}{n_{\ep}^2}\int_{0}^{n_{\ep}(1-\delta)}{\frac{\rho^2}{{1-\rho^2}}\,d\rho}\\&
=-\frac{\ep^{-2\alpha}}{2}\frac{(1-n_{\ep})^2}{n_{\ep}^2}[2n_{\ep}(1-\delta)+\log(1-n_{\ep}(1-\delta))-\log(1+n_{\ep}(1-\delta))].
\end{split}
\end{equation}
Using that $n_{\ep}=1-\frac{\varepsilon^{\alpha}}{|v|^2}+o(\varepsilon^{2\alpha})$, from \eqref{B4} it is clear that 
\begin{equation*}
A_1\simeq -\frac{\ep^{-2\alpha}}{2}(1-n_{\ep})^2(\log(1-n_{\ep}(1-\delta)))=-\frac{\ep^{-2\alpha}}{2}\left(\frac{\ep^{2\alpha}}{|v|^4}\right)\log(\ep^{\alpha}(1-\delta)+\delta).
\end{equation*}
A straightforward computation shows that the right hand side of the previous expression is
\begin{equation}
\label{B4*}
\begin{split}
& 
-\frac{\ep^{-2\alpha}}{2}\frac{\ep^{2\alpha}}{|v|^4}\Big(\log(\ep^{\alpha})+\log(1-\delta+\frac{\delta}{\ep^{\alpha}})\Big)\\&
=-\frac{1}{2|v|^4}\Big(\log(\ep^{\alpha})+\delta(1-\frac{1}{\ep^{\alpha}})\Big)\\&
=\frac{\alpha}{2|v|^4}|\log(\ep)|\Big(1+\frac{\delta(1-\frac{1}{\ep^{\alpha}})}{|\log(\ep^{\alpha})|}\Big).
\end{split}\end{equation}
Choosing $\delta=\frac{\ep^{\alpha}}{|\log{\ep}|^{\gamma}}$ with $\gamma\in(0,\alpha/2)$, it follows $\delta/\ep^{\alpha}\underset{\ep\to 0}{\longrightarrow} 0$. 

In order to compute $A_2$, we need the following estimate for the remainder term
\begin{equation}
|R_1\left(\rho/n_{\ep}\right)|\leq \frac{1}{2} \frac{\rho}{n_{\ep}}\frac{1}{(1-\frac{\rho^2}{n_{\ep}^2})^{\frac{3}{2}}}\left(\frac{\rho}{n_{\ep}}-\rho\right)^2.
\end{equation}
Then
\begin{equation}
\begin{split}
A_2&\leq\frac{\ep^{-2\alpha}}{4}\int_{0}^{n_{\ep}(1-\delta)}{ \frac{\rho^2}{n_{\ep}^2}\frac{1}{(1-\frac{\rho^2}{n_{\ep}^2})^3}\left(\frac{\rho}{n_{\ep}}-\rho\right)^4\,d\rho}
\\&{\underset{u=\frac{\rho}{n_{\ep}}}{=}} \ep^{-2\alpha}n_{\ep}\int_{0}^{1-\delta}{\frac{u^2}{2(1-u^2)^3}u^4(1-n_{\ep})^4\,du}\\&{\underset{v=1-u}{=}}\frac{\ep^{-2\alpha}n_{\ep}}{2}\int_{\delta}^{1}{\frac{(1-v)^6}{v^3}(1-n_{\ep})^4\,dv}\simeq \frac{\ep^{-2\alpha}n_{\ep}}{2}\frac{(1-n_{\ep})^4}{\delta^2}.
\end{split}
\end{equation}
Also in this case, the only significant contribution is given by 
\begin{equation*}
\frac{\ep^{-2\alpha}(1-n_{\ep})^4}{\delta^2}\simeq \frac{\ep^{-2\alpha}\ep^{4\alpha}}{\delta^2}\underset{\ep\to 0}{\longrightarrow} 0
\end{equation*}
again for $\delta=\frac{\ep^{\alpha}}{|\log{\ep}|^{\gamma}}$ with $\gamma\in(0,\alpha/2)$. This shows that 
\begin{equation*}
A=A_1(1+\mathcal O(\ep)).
\end{equation*}

Now we compute $B$ in \eqref{B2}, namely
\begin{equation}
\label{B5}
\begin{split}
&B=\ep^{-2\alpha}\int_{n_{\ep}(1-\delta)}^{n_{\ep}}{[2(\arcsin\left(\rho/n_{\ep}\right)-\arcsin(\rho))]^2\,d\rho}\\&=\ep^{-2\alpha}\int_{n_{\ep}(1-\delta)}^{n_{\ep}}\left(\int_{\rho}^{\frac{\rho}{n_{\ep}}}\,dx \frac{1}{\sqrt{1-x^2}}\right)^2\,d\rho.
\end{split}
\end{equation}
Since
\begin{equation*}
\begin{split}\displaystyle
&\int_{\rho}^{\frac{\rho}{n_{\ep}}}\,dx\frac{1}{\sqrt{1-x^2}}=\int_{\rho}^{\frac{\rho}{n_{\ep}}}\,dx\frac{1}{\sqrt{(1-x)}\sqrt{(1+x)}}\\&
\leq\frac{1}{1+\rho}\int_{\rho}^{\frac{\rho}{n_{\ep}}}\,dx\frac{1}{\sqrt{(1-x)}}\underset{u=1-x}{=}\int_{1-\frac{\rho}{n_{\ep}}}^{1-\rho}\frac{1}{\sqrt{u}}\\&
=\frac{1}{\sqrt{(1-\rho)}}\left(\sqrt{(1-\rho)}-\sqrt{(1-\rho/n_{\ep})}\right)
\end{split}
\end{equation*}
in \eqref{B5} we have
\begin{equation*}
\begin{split}\displaystyle
&B=\ep^{-2\alpha}\int_{n_{\ep}(1-\delta)}^{n_{\ep}}\frac{1}{(1-\rho)}\left(\sqrt{(1-\rho)}-\sqrt{(1-\rho/n_{\ep})}\right)^2\,d\rho\\&
\underset{(1-\frac{\rho}{n_{\ep}}<1-\rho)}{\leq}\ep^{-2\alpha}\int_{n_{\ep}(1-\delta)}^{n_{\ep}}(\rho/n_{\ep}-\rho)\,d\rho\\&=\ep^{-2\alpha}n_{\ep}(1-n_{\ep})[1-(1-\delta)^2]\simeq \ep^{-2\alpha}\ep^{2\alpha}\delta.
\end{split}
\end{equation*}
Again, with the previous choice for $\delta$, this term vanishes in the limit for $\ep\to 0$.

The second integral in the right hand side of \eqref{B2} reads 
\begin{equation}
\begin{split}
&\ep^{-2\alpha}\int_{n_{\ep}}^{1}{\theta^2(\rho)\,d\rho}=\ep^{-2\alpha}\int_{n_{\ep}}^{1}{(\pi-2\arcsin(\rho))^2\,d\rho}
\\&\simeq\ep^{-2\alpha}(1-n_{\ep})^2\simeq \ep^{-2\alpha}\frac{\varepsilon^{2\alpha}}{|v|^2}=\frac 1{|v|^2}.
\end{split}\end{equation}
Therefore the only contribution in the limit is the one given by \eqref{B4*} and we obtain
\begin{equation}
\tilde{B}:=\lim_{\ep\to 0}{\frac{\mu\ep^{-2\alpha}}{2}|v|\int_{-1}^{1}{\theta^2(\rho)\,d\rho}}=\lim_{\ep\to 0}\mu\left[\frac{2\alpha}{|v|^3}|\log(\ep)|\right]=+\infty,
\end{equation}
and finally
\begin{equation*}
B:=\lim_{\ep\to 0}{\frac{\mu\ep^{-2\alpha}}{2|\log\ep|}|v|\int_{-1}^{1}{\theta^2(\rho)\,d\rho}}=\frac{2\alpha}{|v|^3}\mu.
\end{equation*}


\end{document}